\documentclass[12pt]{article}
\usepackage{amsfonts}
\usepackage{amsthm,amsmath,indentfirst,amssymb}
\usepackage{algorithm} %format of the algorithm
\usepackage{algorithmic} %format of the algorithm
\newtheorem{lemma}{Lemma}[section]

\newtheorem{definition}[lemma]{Definition}

\newtheorem{Theorem}[lemma]{Theorem}
\newtheorem{example}[lemma]{Example}
\newtheorem{remark}[lemma]{Remark}

\begin{document}
\title{Randomized Bicriteria Approximation Algorithm for Minimum Submodular Cost Partial Multi-Cover Problem }
%\thanks{ This research is supported by NSFC (11531011,61222201)}

\author{Yishuo~Shi$^1$, Zhao~Zhang$^2$\thanks{Corresponding Author: Zhao Zhang, hxhzz@sina.com.}, Ding-Zhu~Du$^3$\\% <-this % stops a space
{\small $^1$ College of Mathematics and System Sciences, Xinjiang University}\\
{\small Urumqi, Xinjiang, 830046, China.}\\
{\small $^2$ College of Mathematics Physics and Information Engineering, Zhejiang Normal University}\\
{\small Jinhua, Zhejiang, 321004, China.}\\
{\small Department of Computer Science, University of Texas at Dallas}\\
{\small Richardson, Texas, 75080, USA.}\\
}

\date{}
\maketitle

\begin{abstract}
This paper studies randomized approximation algorithm for a variant of the set cover problem called {\em minimum submodular cost partial multi-cover} (SCPMC).

In a \emph{partial set cover problem}, the goal is to find a minimum cost sub-collection of sets covering at least a required fraction of elements. In a {\em multi-cover problem}, each element $e$ has a covering requirement $r_e$, and the goal is to find a minimum cost sub-collection of sets $\mathcal S'$ which fully covers all elements, where an element $e$ is fully covered by $\mathcal S'$ if $e$ belongs to at least $r_e$ sets of $\mathcal S'$. In a \emph{minimum submodular cost set cover problem} (SCSC), the cost function on sub-collection of sets is submodular and the goal is to find a set cover with the minimum cost.

The SCPMC problem studied in this paper is a combination of the above three problems, in which the cost function on sub-collection of sets is submodular and the goal is to find a minimum cost sub-collection of sets which fully covers at least $q$-percentage of all elements. Previous work shows that such a combination enormously increases the difficulty of studies, even when the cost function is linear.

In this paper, assuming that the maximum covering requirement $r_{\max}=\max_e r_e$ is a constant and the cost function is nonnegative, monotone nondecreasing, and submodular, we give the first randomized bicriteria algorithm for SCPMC the output of which fully covers at least $(q-\varepsilon)$-percentage of all elements and the performance ratio is $O(b/\varepsilon)$ with a high probability, where $b=\max_e\binom{f}{r_{e}}$ and $f$ is the maximum number of sets containing a common element. The algorithm is based on a novel non-linear program. Furthermore, in the case when the covering requirement $r\equiv 1$, a bicriteria $O(f/\varepsilon)$-approximation can be achieved even when monotonicity requirement is dropped off from the cost function.

\end{abstract}

\section{Introduction}\label{sec1}
{\em Set Cover} is one of the most important combinatorial optimization problems in both the theoretical field and the application field, the goal of which is to find a sub-collection of sets with the minimum cost to cover all elements. There are a lot of variants of the set cover problem. The minimum \emph{partial set cover} problem (PSC) is to find a minimum cost sub-collection of sets to cover at least $q$-percentage of all elements. One motivation of PSC comes from the phenomenon that in a real world, ``satisfying all requirements'' will be too costly or even impossible, because of resource limitation or political policy. Another variant is the minimum \emph{multi-cover} problem (MC), which comes from the requirement of fault tolerance in practice. In MC, each element $e$ has a covering requirement $r_e$, and the goal is to find a minimum cost sub-collection $\mathcal S'$ to fully cover all elements, where element $e$ is fully covered by $\mathcal S'$ if $e$ belongs to at least $r_e$ sets of $\mathcal S'$. Another generalization of set cover is {\em submodular cost set cover} (SCSC), in which the cost function on sub-collection of sets is submodular and the goal is to find a set cover with the minimum cost. Submodular functions have a natural diminishing returns property which finds wide applications in the real world, including economics, game theory, machine learning and computer vision, etc.

In this paper, we consider a problem which is a combination of the above three problems. In the \emph{minimum submodular cost partial multi-cover} problem (SCPMC), each element has a profit as well as a covering requirement, the goal is to find a minimum submodular cost sub-collection of sets such that the profit of fully covered elements is at least a fixed percentage of the total profit.

\subsection{Related Work}

For Set Cover, Hochbaum \cite{Hochbaum} gave an $f$--approximation algorithm based on LP rounding where $f$ is the maximum number of sets containing a common element. Khot and Regev \cite{Khot} showed  that the set cover problem cannot be approximated within $f-\varepsilon$ for any constant $\varepsilon>0$ assuming that unique games conjecture is true. Another classic result on Set Cover is that greedy strategy yields a $\ln\Delta$-approximation \cite{Chvatal,Johnson,Lovasz1}, where $\Delta$ is the maximum cardinality of a set. Dinur and Steurer \cite{Dinur} showed that the set cover problem cannot be approximated to $(1-o(1))\ln n$ unless $P=NP$, where $n$ is the size of ground set.

For MC, Dobson \cite{Dobson} gave an $H_{K}$-approximation algorithm for the {\em minimum multi-set multi-cover problem} (MSMC), where $K$ is the maximum size of a multi-set and $H_K=\sum_{i=1}^K1/i$ is the harmonic number (recall that $H_K\approx \ln K$). Rajagopalan and Vazirani \cite{Rajagopalan} gave a greedy algorithm achieving the same performance ratio, using dual fitting analysis. For the {\em minimum set $k$-cover problem} in which the covering requirement of every element is $k$, Berman \emph{et al.} \cite{Berman} gave a randomized algorithm achieving expected performance ratio at most $\ln(\frac{\Delta}{k})$.

For PSC, Kearns \cite{Kearns} gave the first greedy algorithm achieving performance ratio $(2H_n+3)$. Refining the greedy algorithm, Slavik \cite{Slavik} improved the ratio to $H_{\min \{\lceil qn\rceil,\Delta\}}$, where $q$ is the desired covering ratio. Using primal dual method, Gandhi \emph{et al.} \cite{Gandhi} obtained an $f$-approximation. Bar-Yehuda \cite{Bar-Yehuda} studied a generalized version of the partial cover problem in which each element has a profit. Using local ratio method, he also obtained an $f$-approximation. Proposing an Lagrangian relaxation framework, Konemann \emph{et al.} \cite{Konemann} gave a $(\frac{4}{3}+\varepsilon)H_{\Delta}$-approximation for the generalized partial cover problem.

From the above related work, it can be seen that both PSC and MC admit performance ratios which match those best ratios for the classic set cover problem. However, combining partial cover with multi-cover seems to enormously increase the difficulty of studies. Ran \emph{et al.} \cite{Ran} were the first to study approximation algorithm for the {\em minimum partial multi-cover problem} (PMC). Using greedy strategy and a delicate dual fitting analysis, they gave a $\gamma H_{\Delta}$-approximation algorithm, where $\gamma=1/(1-(1-q)\eta)$, $\eta=\Delta\frac{c_{\max}}{c_{\min}}\frac{r_{\max}}{r_{\min}}$, and $c_{\max}$, $c_{\min}$ are the maximum and the minimum cost of set, $r_{\max}$, $r_{\min}$ are the maximum and the minimum covering requirement of element, respectively. This ratio is meaningful only when the covering percentage $q$ is very close to 1. In \cite{Ran1}, Ran {\it et al.} presented a simple greedy algorithm achieving performance ratio $\Delta$. Recall that in terms of $\Delta$, greedy algorithm for Set Cover achieves performance ratio $\ln \Delta$. So, ratio $\Delta$ for PMC is exponentially larger than the one for Set Cover. In the same paper, they also presented a local ratio algorithm which reveals an interesting ``shock wave'' phenomenon: their performance ratio is $f$ for both PSC (that is, when $r_{\max}=r_{min}=1$ which is the partial {\em single} cover problem) and MC (that is, when $q=1$ which is the {\em full} multi-cover problem); however, when $q$ is smaller than 1 by a very small constant, the ratio jumps abruptly to $O(n)$.

The submodular cost set cover problem was first proposed by Iwata and Nagano \cite{Iwata}. They gave an $f$-approximation algorithm for nonnegative submodular functions. In paper \cite{Koufogiannakis}, Koufogiannakis and Young generalized set cover constraint to arbitrary covering constraints and gave an
$f$-approximation algorithm for monotone nondecreasing nonnegative submodular functions.

In this paper we combine submodular cost function with partial multi-cover constraint. As one can see from previous results on PMC, even when the cost function is linear, the partial multi-cover problem is already very difficult.

\subsection{Our Contribution}

%There are a lot of works studying the minimum submodular cost {\em full} cover problem. However, for the combination of partial cover and multi-cover, even when the object function is linear, the problem is very hard, defeating those classic methods which have successfully achieved good approximation for submoduar cost full cover.

The major contribution of this paper is a randomized $(\varepsilon,O(\frac{b}{\varepsilon}))$-approximation algorithm for SCPMC, that is, the algorithm produces a solution covering at least $(q-\varepsilon)$-percentage of the total covering requirement, and achieves performance ratio $O(\frac{b}{\varepsilon})$ with a high probability, where $b=\max_e\binom{f}{r_{e}}$, and $f$ is the maximum number of sets containing a common element.

Before presenting this algorithm, we show that a natural integer program for SCPMC does not work since its integrality gap is arbitrarily large. Hence, to obtain a good approximation, we propose a novel integer program. The relaxation of the integer program uses Lov\'{a}sz extension \cite{Lovasz}. Our algorithm consists of two stages of rounding. The first stage is a deterministic rounding. The second stage is a random rounding, the analysis of which is based on an equivalent expression of Lov\'{a}sz extension \cite{Chekuri} in view of expectation.

As far as we know, this is the first approximation algorithm for a partial version of the submodular multi-cover problem. Furthermore, we show that for the special case when the covering requirement $r\equiv 1$ (the special case is abbreviated as SCPSC), our method can be adapted to yield an $(\varepsilon,O(f/\varepsilon))$-approximation with high probability, even when monotonicity is dropped off from the requirement of the cost function.

This paper is organized as follows. In Section 2, we introduce formal definitions of problems considered in this paper, as well as some technical results. The bicriteria randomized algorithm for SCPMC is presented and analyzed in Section 3. In Section 4, we show how to adapt our algorithm to deal with SCPSC. The last section concludes the paper and discusses some future work.

\section{Preliminaries}
\begin{definition}[Submodular Cost Partial Multi-Cover (SCPMC)]
{\rm Suppose $E$ is an element set and $\mathcal{S}\subseteq 2^{E}$ is a collection of subsets of $E$ with $\bigcup_{S\in\mathcal S}S=E$; each element $e\in E$ has a positive covering requirement $r_e$ and a positive profit $p_e$; cost function $\rho_0: 2^{\mathcal S}\mapsto\mathbb{R}$ is defined on sub-collections of $\mathcal S$, which is nonnegative, monotone nondecreasing, and submodular. Given a constant $q\in (0,1]$ called {\em covering ratio}, the SCPMC problem is to find a minimum cost sub-collection $\mathcal S'$ such that $\sum_{e\sim\mathcal S'} p(e)\geq qP$, where $P=\sum_{e\in E}p(e)$ is the total profit, $e\sim S'$ means that $e$ is {\em fully covered} by $\mathcal S'$, that is, $|\{S\in\mathcal S'\colon e\in S\}|\geq r_e$. An instance of SCPMC is denoted as $(E,\mathcal S,r,p,q,\rho_0)$.}
\end{definition}

In particular, when $r_{\max}=1$, we call the problem a \emph{submodular cost partial set cover problem} (SCPSC). When the cost function is linear, that is, every set $S\in \mathcal S$ has a cost $c(S)$ and the cost of a sub-collection $\mathcal S'$ is $\rho_0(\mathcal S')=\sum_{S\in\mathcal S'}c(S)$, the problem is exactly the minimum {\em partial multi-cover problem} (PMC).

Submodular function has many equivalent definitions. We only introduce the following one which is convenient to be used in this paper.

\begin{definition}[submodular function]\label{def2.2}
{\rm Given a ground set $E$, a set function $\rho: 2^E\mapsto \mathbb{R}$ is {\em submodular} if for any $E''\subseteq E'\subseteq E$ and $E_0\subseteq E\setminus E'$, we have
\begin{equation}\label{eq11-22-1}
  \rho(E'\cup E_0)-\rho(E')\leq \rho(E''\cup E_0)-\rho(E'').
\end{equation}}
\end{definition}
Notice that a nonnegative submodular function $\rho$ satisfies {\em subadditivity}: for any sets $X,Y\subseteq E$,
\begin{equation}\label{eq12-28-1}
  \rho(X\cup Y)\leq \rho(X)+\rho(Y).
\end{equation}

Notice that a set $S\subseteq E$ can be indicated by its characteristic vector $x_S=(x_1,\ldots,x_n)$, where $n=|E|$, $E=\{e_1,\ldots,e_n\}$, and $x_i=1$ if $e_i\in S$ and $x_i=0$ if $e_i\notin S$. So, in the following, we shall use notation $\{0,1\}^n\mapsto \mathbb R$ to refer to a set function. The relationship between submodularity and convexity can be formulated in terms of Lov\'{a}sz extension.

\begin{definition}[Lov\'{a}sz extension \cite{Lovasz}]\label{def1.3}
{\rm For a set function $\rho: \{0,1\}^n\mapsto \mathbb{R}$, the Lov\'{a}sz extension ${\hat{\rho}}: \mathbb R^n\rightarrow \mathbb{R}$ is defined as follows. For any vector $x\in \mathbb R^n$, order elements as $e_{j_1},e_{j_2},\ldots,e_{j_n}$ such that $x_{j_1}\geq x_{j_2}\geq ...\geq x_{j_n}$, where $x_{j_i}$ is the coordinate of $x$ indexed by $e_{j_i}$.
Let $E_i=\{e_{j_1},e_{j_2},...,e_{j_i}\}$. The value of ${\hat{\rho}}$ at $x$ is
\begin{equation}\label{eq1}
  {\hat{\rho}}(x)=\sum^{n-1}_{i=1}(x_{j_i}-x_{j_{i+1}})\rho(E_i)+x_{j_n}\rho(E_n).
\end{equation}}
\end{definition}

The above definition implies that Lov\'{a}sz extension $\hat{\rho}$ satisfies {\em positive homogenous property}, that is, for any $t>0, \hat{\rho}(tx)=t\hat{\rho}(x)$. The following result reveals the relationship between submodularity and convexity.

\begin{Theorem}\label{th1}
A set function $\rho$ is submodular if and only if its ${\rm Lov\acute{a}sz}$ extension $\hat{\rho}$ is convex.
\end{Theorem}

The following is an equivalent expression of Lov\'{a}sz extension in range $[0,1]^n$.

\begin{Theorem}[\cite{Chekuri}]\label{def1.4}
{\rm Let $\rho$ be a set function $\{0,1\}^n\mapsto \mathbb{R}$. The Lov\'{a}sz extension ${\hat{\rho}}$ of $\rho$ in range $[0,1]^n$ can be equivalently expressed as
\begin{equation}\label{eq2}
  \hat{\rho}(x) =\mathop{\mathbb{E}}\limits_{\theta \in [0,1]}[\rho(x^\theta)]=\int_{0}^{1} \rho(x^{\theta})d\theta,
\end{equation}
where $x^{\theta}_i=1$ if $x_i\geq \theta$, otherwise $x^{\theta}_i=0$.}
\end{Theorem}

In this paper, we study the SCPMC problem under the following assumptions.

{\bf (Assumption 1)} The maximum covering requirement $r_{\max}=\max\{r_e\colon e\in E\}$ has a constant upper bound.

{\bf (Assumption 2)} Since submodular cost (full) multi-cover problem is already studied in \cite{Iwata,Koufogiannakis}, we only consider the partial version, that is, it is assumed that $q<1$.

%{\bf (Assumption 2)} For each $e$,  $p(e)\leq qP$. In fact, if this is not true, then one can obtain the same ratio by the following method. For every element $e$ with $p(e)>qP$, let $\mathcal S_e$ be a family of $r_e$ least-cost sets containing $e$. A desired solution can be obtained by comparing the output of the following algorithm with the costs of these $\mathcal S_e$'s and choosing the best one.

\section{Approximation Algorithm for SCPMC}
A natural idea to model the SCPMC problem is to use the following integer programm:
\begin{align}\label{eq12-22-1}
\min\ & \rho_0(x)  \nonumber\\
 s.t.\    \sum_{e:\ e\in E}p_ey_e  & \geq qP, \nonumber\\
 \sum_{S:\ e\in S}x_S  & \geq r_{e}y_e, \ \mbox{for any}\ e\in E\\
   x_S \in  \{0,1\} & \ \ \mbox{for}\ S\in\mathcal S\nonumber\\
   y_e \in  \{0,1\} & \ \ \mbox{for}\ e\in E\nonumber
\end{align}
Here $x_S$ indicates whether set $S$ is selected and $y_e$ indicates whether element $e$ is fully covered. The second constraint says that if $y_e=1$ then at least $r_e$ sets containing $e$ must be selected and thus $e$ is fully covered. Relaxing \eqref{eq12-22-1}, we have the following convex program:
\begin{align}\label{eq12-22-2}
\min\ & \hat{\rho}_0(x)  \nonumber\\
 s.t.\    \sum_{e:\ e\in E}p_ey_e  & \geq qP, \nonumber\\
 \sum_{S:\ e\in S}x_S  & \geq r_{e}y_e,\ \mbox{for any}\ e\in E\\
    x_S& \geq 0\ \mbox{for}\ S\in\mathcal S\nonumber\\
  1\geq y_e& \geq 0\ \mbox{for}\ e\in E\nonumber
\end{align}

However, based on such a program, one cannot find a good approximation. The following example shows that the integrality gap between \eqref{eq12-22-1} and \eqref{eq12-22-2} can be arbitrarily large, even when the profit function is a constant and the cost function is linear.

\begin{example}\label{ex1}
{\rm Let $E=\{e_1,e_2\}$, $\mathcal S=\{S_1,S_2,S_3\}$ with $S_1=\{e_1\}$, $S_2=\{e_2\}$, $S_3=\{e_1,e_2\}$, $c(S_1)=c(S_2)=1$, $c(S_3)=M$ where $M$ is a large positive number, $r(e_1)=r(e_2)=2$, $p(e_1)=p(e_2)=1$,  $q=1/2$, and the cost function $\rho_0(x)=\sum_{S\in\mathcal S}c(S)x_S$. Then $x_{S_1}=x_{S_2}=1$, $x_{S_3}=0$, $y_{e_1}=y_{e_2}=1/2$ form a feasible solution to \eqref{eq12-22-2} with objective value 2, while any integral feasible solution to \eqref{eq12-22-1} has cost at least $M+1$.}
\end{example}

Hence, to obtain a good approximation, we need to find another program.

\subsection{Integer Program and Convex Relaxation}
For an element $e$, an {\em $r_e$-cover} is a sub-collection $\mathcal A\subseteq \mathcal S$ with $|\mathcal A|=r_e$ such that $e\in S$ for every $S\in\mathcal A$. Denote by $\Omega_e$ the family of all $r_e$-covers and $\Omega=\bigcup_{e\in E}\Omega_e$. The following example illustrates these concepts.

\begin{example}\label{ex2}
{\rm Let $E=\{e_1,e_2,e_3\}$. $\mathcal S=\{S_1,S_2,S_3\}$ with $S_1=\{e_1,e_2\}$, $S_2=\{e_1,e_2,e_3\}$, $S_3=\{e_2,e_3\}$, $S_4=\{e_1,e_3\}$, and $r(e_1)=2$, $r(e_2)=r(e_3)=1$. For this example, $\Omega_{e_1}=\{\{S_1,S_2\},\{S_1,S_4\},\{S_2,S_4\}\}, \Omega_{e_2}=\{\{S_1\},\{S_2\},\{S_3\}\}$, $\Omega_{e_3}=\{\{S_2\},\{S_3\},\{S_4\}\}$, and $\Omega=\{\{S_1\},\{S_2\},\{S_3\},\{S_4\},\{S_1,S_2\},\{S_1,S_4\},\{S_2,S_4\}\}$.}
\end{example}

Let $\rho$: $2^{\Omega}\rightarrow \mathbb{R}$ be the function on sub-families of $\Omega$ defined by
\begin{equation}\label{eq17-1-3-1}
\rho(\Omega')=\rho_0( \bigcup_{\mathcal A\in \Omega'}\mathcal A)
\end{equation}
for $\Omega'\subseteq \Omega$. For example, $\rho(\{\{S_1\},\{S_1,S_2\}\})=\rho_0(\{S_1,S_2\})$. The SCPMC problem can be modeled as an integer program as follows.
\begin{align}\label{eq11-22-2}
\min\ \rho(x)  \nonumber\\
 s.t.\    \sum_{e:\ e\in E}p_ey_e  & \geq qP, \nonumber\\
 \sum_{\mathcal A:\ \mathcal A\in \Omega_{e}}x_{\mathcal A}  & \geq y_e, \ \mbox{for any}\ e\in E\\
   x_{\mathcal A} \in  \{0,1\} & \ \ \mbox{for}\ \mathcal A\in\Omega\nonumber\\
   y_e \in  \{0,1\} & \ \ \mbox{for}\ e\in E\nonumber
\end{align}
Here, $x_{\mathcal A}$ indicates whether cover $\mathcal A$ is selected and $y_e$ indicates whether element $e$ is fully covered. The second constraint says that if $y_e=1$, then at least one $r_e$-cover must be selected and thus $e$ is fully covered.

\begin{example}\label{ex3}
{\rm For the example in Example \ref{ex2}, suppose $p_{e_i}\equiv 1$ for $i=1,2,3$ and $q=2/3$. Consider a feasible solution to \eqref{eq11-22-2}: $x_{\mathcal A_1}=x_{\mathcal A_2}=1$ for $\mathcal A_1=\{S_1,S_2\}$, $\mathcal A_2=\{S_2\}$, and $x_{\mathcal A}=0$ for all other $\mathcal A\in\Omega\setminus \{\mathcal A_1,\mathcal A_2\}$, we have $y_{e_1}=y_{e_2}=1$ and $y_{e_3}=0$. This feasible solution to \eqref{eq11-22-2} has objective value $\rho(\{\mathcal A_1,\mathcal A_2\})=\rho_0(S_1,S_2)$, which corresponds to a feasible solution $\{S_1,S_2\}$ to SCPMC with the same cost. Conversely, for the feasible solution $\{S_1,S_2\}$ to SCPMC, it is natural to set $x_{\mathcal A_1}=1$ and all other $x_{\mathcal A}$ to be zeros. However, this is not a feasible solution to \eqref{eq11-22-2}. Nevertheless, one can construct a feasible solution to \eqref{eq11-22-2} having the same cost by setting $x_{\mathcal A_1}=x_{\mathcal A_2}=1$ and all other $x_{\mathcal A}$ to be zeros.}
\end{example}

In general, for a feasible solution $\mathcal S'$ to SCPMC, one can construct a feasible solution to \eqref{eq11-22-2} as follows: for each element $e$ which is fully covered by $\mathcal S'$, let $y_e=1$ and let $x_{\mathcal A_e}=1$ for exactly one $r_e$-cover $\mathcal A_e$ which contains $r_e$ subsets of $\mathcal S'$ (such $\mathcal A_e$ exists since $e$ is fully covered by $\mathcal S'$); all other variables are set to be zeros. Such a construction clearly results in a feasible solution to \eqref{eq11-22-2} whose objective value is at most $\rho_0(\mathcal S')$ (by the monotonicity of $\rho_0$). So, \eqref{eq11-22-2} is indeed a characterization of the SCPMC problem.

The following lemma shows that function $\rho$ is nonnegative, monotone nondecreasing, and submodular.

\begin{lemma}\label{lem1-15-1}
If $\rho_0$ is nonnegative, monotone nondecreasing, and submodular, then the function $\rho$ defined in \eqref{eq17-1-3-1} is also nonnegative, monotone nondecreasing, and submodular.
\end{lemma}
\begin{proof}
The nonnegativity and the monotonicity are obvious. To prove the submodularity, by Definition \ref{def2.2}, it is sufficient to show that for any $\Omega''\subseteq \Omega'\subseteq \Omega$ and $\Omega_0\subseteq \Omega\setminus \Omega'$,
\begin{equation}\label{eq17-1-3-3}
\rho(\Omega'\cup \Omega_0)-\rho(\Omega')\leq \rho(\Omega''\cup \Omega_0)-\rho(\Omega'').
\end{equation}

Denote $\bigcup_{\mathcal A\in \Omega'}\mathcal A=\mathcal S'$ and $\bigcup_{\mathcal A\in \Omega''}\mathcal A=\mathcal S''$. Since $\Omega''\subseteq \Omega'$, we have $\mathcal S''\subseteq \mathcal S'$. Denote $\mathcal S_1=\left(\bigcup_{\mathcal A\in \Omega'\cup\Omega_0}\mathcal A\right)\setminus \mathcal S'$ and $\mathcal S_2=\left(\bigcup_{\mathcal A\in \Omega''\cup\Omega_0}\mathcal A\right)\setminus \mathcal S''$. Then $\mathcal S_1\subseteq \mathcal S_2$. Combining this with the observation that $\mathcal S'\cup\mathcal S_1=\bigcup_{\mathcal A\in \Omega'\cup\Omega_0}\mathcal A\supseteq \bigcup_{\mathcal A\in \Omega''\cup\Omega_0}\mathcal A=\mathcal S''\cup \mathcal S_2$, we have
\begin{equation}\label{eq17-1-3-2}
\mathcal S''\subseteq (\mathcal S''\cup\mathcal S_2)\setminus \mathcal S_1\subseteq \mathcal S'.
\end{equation}
It follows that
\begin{align*}
  \rho(\Omega'\cup \Omega_0)-\rho(\Omega') & = \rho_0(\mathcal S'\cup \mathcal S_1)-\rho_0(\mathcal S')\\
  & \leq \rho_0(((\mathcal S''\cup \mathcal S_2)\setminus \mathcal S_1)\cup \mathcal S_1)-\rho_0((\mathcal S''\cup \mathcal S_2)\setminus \mathcal S_1)\\
     &\leq \rho_0(((\mathcal S''\cup \mathcal S_2)\setminus \mathcal S_1)\cup \mathcal S_1)-\rho_0(\mathcal S'')\\
     & =\rho_0(\mathcal S''\cup S_2)-\rho_0(\mathcal S'')\\ &=\rho(\Omega''\cup \Omega_0)-\rho(\Omega''),
\end{align*}
where the first inequality uses submodularity of $\rho_0$ and \eqref{eq17-1-3-2}, and the second inequality uses the monotonicity of $\rho_0$ and \eqref{eq17-1-3-2}. Inequality \eqref{eq17-1-3-3}, and thus the lemma, is proved.
\end{proof}

\begin{remark}\label{rem1-15-2}
{\rm If $\rho_0$ is nonnegative and submodular but is not monotone nondecreasing, then $\rho$ is not necessarily submodular. Consider the following example. Let $\mathcal S=\{S_1,S_2,S_3\}$ with $\rho_0(\{S_1\})=\rho_0(\{S_1,S_3\})=1$ and $\rho_0(\mathcal S')=0$ for any other sub-collection $\mathcal S'\subseteq \mathcal S$. It can be verified that $\rho_0$ is nonnegative and submodular. Consider sub-families $\Omega''=\{\{S_1\}\} \subseteq \Omega'=\{\{S_1\},\{S_1,S_2\}\}$ and $\Omega_0=\{\{S_1,S_2,S_3\}\}$, it can be calculated that
$$\rho(\Omega'\cup \Omega_0)-\rho(\Omega')=0-0=0>-1=0-1=\rho(\Omega''\cup \Omega_0)-\rho(\Omega'').$$
So, $\rho$ is not submodular.}
\end{remark}

Let $\hat{\rho}$ be the Lov\'{a}sz extension of $\rho$. By Theorem \ref{th1}, $\hat{\rho}$ is convex. Relaxing \eqref{eq11-22-2}, we have the following convex program:
\begin{align}\label{eq11-22-3}
  \min\ & \hat{\rho}(x)  \nonumber\\
 s.t.\    \sum_{e:\ e\in E}p_ey_e  & \geq qP, \nonumber\\
 \sum_{\mathcal A:\ \mathcal A\in \Omega_{e}}x_{\mathcal A}  & \geq y_e, \ \mbox{for any}\ e\in E\\
   x_{\mathcal A}& \geq 0\ \mbox{for}\ \mathcal A\in\Omega \nonumber\\
  1\geq y_e& \geq 0\ \mbox{for}\ e\in E\nonumber
\end{align}

\begin{lemma}
Convex program \eqref{eq11-22-3} is polynomial-time solvable.
\end{lemma}
\begin{proof}
It is known that (see \cite{Dughmi}) for a submodular function $\rho$, its Lov'asz extension $\hat \rho (x)=\rho^-(x)$ for any $x\in [0,1]^{|\Omega|}$, where $\rho^-$ is the {\em convex closure} of $\rho$ defined as follows. For each sub-family $\Omega'$ of $\Omega$, denote by $\chi_{\Omega'}$ as the indicator vector of $\Omega'$. The convex closure of $\rho$ is the function $\rho^-$: $[0,1]^{|\Omega|}\mapsto \mathbb{R}$ such that for any vector $x\in [0,1]^{|\Omega|}$, $\rho^-(x)=\min \{\sum_{\Omega'\subseteq \Omega}\lambda_{\Omega'}\rho(\Omega'): \sum_{\Omega'\subseteq \Omega}\lambda_{\Omega'}\chi_{\Omega'}=x, \sum_{\Omega'\subseteq \Omega}\lambda_{\Omega'}=1, \lambda_{\Omega'}\geq 0  \}$. Hence $\eqref{eq11-22-3}$ can be rewritten as:
\begin{align}\label{eq11-22-5}
  \min\ & \sum_{\Omega'\subseteq \Omega}\lambda_{\Omega'}\rho(\Omega')  \nonumber\\
s.t.\ \ \ & \sum_{\Omega'\subseteq \Omega}\lambda_{\Omega'} =1, \nonumber\\
 &    \sum_{e:\ e\in E}p_ey_e   \geq qP, \nonumber\\
 \sum_{\Omega': \mathcal A\in \Omega'\subseteq \Omega}\lambda_{\Omega'}& =x_{\mathcal A}, \ \mbox{for any}\ \mathcal A\in \Omega \nonumber\\
 \sum_{\mathcal A:\ \mathcal A\in \Omega_{e}}x_{\mathcal A}  & \geq y_e, \ \mbox{for any}\ e\in E\\
     \lambda_{\Omega'}& \geq 0\ \mbox{for}\ \Omega'\subseteq\Omega\nonumber\\
   x_{\mathcal A}& \geq 0\ \mbox{for}\ \mathcal A\in\Omega\nonumber\\
  1\geq y_e& \geq 0\ \mbox{for}\ e\in E\nonumber
\end{align}
Notice that this is a linear program. For each element $e$, $|\Omega_e|\leq b=\max_e\binom{f}{r_{e}}$. Since in Assumption 1, we have assumed that $r_{\max}$ is upper bounded by a constant, the number of variables in the form of $x_{\mathcal A}$ or $y_e$ is polynomial. However, the number of variables in the form of $\lambda_{\Omega'}$ is exponential.

Consider the dual program of \eqref{eq11-22-5}:
\begin{align}\label{eq11-22-6}
 & \max\  a+bqP-\sum_{e\in E}f_e  \nonumber\\
s.t.\ \ \ & a+\sum_{\mathcal A\in \Omega'}c_{\mathcal A}\leq \rho(\Omega'), \ \mbox{for any}\ \Omega'\subseteq\Omega\\
 &    \sum_{e:\ e\in \mathcal A}d_e-c_{\mathcal A} \leq 0,  \ \mbox{for any}\ \mathcal A\in \Omega \nonumber\\
 & p_eb-d_e-f_e \leq 0, \ \mbox{for any}\ e\in E \nonumber\\
 & b\geq 0\ \mbox{and}\ d_e, f_e \geq 0 \ \mbox{for}\ e\in E\nonumber
\end{align}
Since both $|\Omega|$ and $|E|$ are polynomial, to solve \eqref{eq11-22-6}, it suffices to construct a separation oracle for the first set of constraints.

Define $g(\Omega')=\rho(\Omega')-\sum_{\mathcal A\in \Omega'}c_{\mathcal A}$ for any $\Omega'\subseteq\Omega$. Since $g$ is obtained by subtracting a modular function from a submodular function, $g$ is also a submodular function. Hence, by finding a minimizer of $g$, which can be done in polynomial time, and then check whether its $g$-value is at least $a$, we can either claim the validity of the first set of constraints or find out a violated constraint.
\end{proof}

Since \eqref{eq11-22-3} is a relaxation of \eqref{eq11-22-2}, we have $opt_{cp}\leq opt$, where $opt_{cp}$ is the optimal value of \eqref{eq11-22-3} and $opt$ is the optimal integer value of \eqref{eq11-22-2} (which is also the optimal value of SCPMC).

\subsection{Rounding Algorithm}

For a sub-collection $\mathcal S'\subseteq \mathcal S$, denote by $\mathcal C(\mathcal S')$ the set of elements fully covered by $\mathcal S'$. Two parameters $s,t$ are needed which are chosen in Theorem \ref{th2} to guarantee the desired ratio with high probability. %are constants satisfying
%\begin{equation}\label{eq11-22-4}
%s>t, 1<t\leq 1+\frac{1-q}{q'}\ \mbox{and}\ 0<q'\leq q<1
%\end{equation}
The rounding algorithm consists of two phases. In the first phase, a deterministic rounding is executed to form a sub-collection $\mathcal S_1$. In the second phase, a randomized rounding is executed to form a sub-collection $\mathcal S_2$. The output is the union of $\mathcal S_1$ and $\mathcal S_2$.

\begin{algorithm}
\caption{Algorithm for $SCPMC$}
{\bf Input:} A $SCPMC$ instance $(E,\mathcal S,r,p,q,\rho_0)$, two parameters $s,t$ satisfying $1<t<s\leq 1/q$, and a real positive number $\varepsilon<q$.

{\bf Output:} A sub-collection $\mathcal S'$ which has total covering profit at least $(q-\varepsilon)P$.

\begin{algorithmic}[1]\label{alg1}
\STATE Find an optimal solution $(x^{*},y^{*})$ to \eqref{eq11-22-3}.
\STATE $\mathcal S_{1}\leftarrow \emptyset$, $\mathcal S_{2}\leftarrow\emptyset$.
\FOR {all $e$ with $y^{*}_e\geq \frac{1}{s}$}
    \STATE For each $\mathcal A\in \Omega_e$ with $x^{*}_{\mathcal A}\geq \frac{1}{bs}$, let $\hat x_{\mathcal A}\leftarrow 1$.
\ENDFOR
\STATE For all $x^*_{\mathcal A}$ which is not rounded up to $1$, set $\hat x_{\mathcal A}\leftarrow 0$.
\STATE $\mathcal S_1\leftarrow\{S\colon S\in \mathcal A\ \mbox{with}\ \hat x_{\mathcal A}=1\}$.
\STATE If $\mathcal S_1$ has total covering profit at least $(q-\varepsilon)P$ then output $\mathcal S'\leftarrow \mathcal S_1$ and stop.
\STATE $E'\leftarrow E-\mathcal C(\mathcal S_1)$, $q'\leftarrow (qP-p(\mathcal C(\mathcal S_1)))/P$.
\FOR {$i=1$ to $s\ln(\frac{s}{s-t})b$ }
    \STATE Pick $\theta \in [0,1]$ randomly uniformly.
    \STATE For each remaining $\mathcal A$ with $x^*_{\mathcal A}\geq \theta$, set $\hat x_{\mathcal A}\leftarrow 1$ and $\mathcal S_2\leftarrow \mathcal S_2\cup \{S\colon S\in \mathcal A\}$.
\ENDFOR

\STATE Output $\mathcal S'=\mathcal S_2\cup \mathcal S_2$.

\end{algorithmic}\label{algo1}
\end{algorithm}

\subsection{Approximation Analysis}
\begin{lemma}\label{lem1}
For the collection of sets $\mathcal S_1$ computed by Algorithm \ref{algo1}, $\rho_0(\mathcal S_1)\leq bs\cdot opt_{cp}$. Furthermore, all elements with $y^*_e\geq \frac{1}{s}$ are fully covered by $\mathcal S_1$.
\end{lemma}
\begin{proof}
Let $\hat x$ be the vector defined after Line 6 of Algorithm \ref{algo1}, and let $z$ be the vector with $z_{\mathcal A}=\min\{1,bsx^*_{\mathcal A}\}$ for $\mathcal A\in \Omega$.

Recall that Lov\'asz extension in Definition \ref{def1.3} requires an ordering of elements in a non-increasing manner. By the definition of $z$ and by the nonnegativity of $\rho$, we can take the ordering of elements defining $\hat{\rho}(z)$ and $\hat{\rho}(bsx^*)$ to be the same and
\begin{equation}\label{eq11-29-1}
\hat\rho(z)\leq \hat\rho(bsx^*).
\end{equation}

We claim that $\hat x_{\mathcal A}\leq z_{\mathcal A}$ holds for any index $\mathcal A\in\Omega$. This is clearly true if $\hat x_{\mathcal A}=0$. For an index $\mathcal A$ with $\hat x_{\mathcal A}=1$, we have $x_{\mathcal A}^*\geq 1/bs$ (by Line 4 of Algorithm \ref{algo1}), which implies $z_{\mathcal A}=1$. The claim is proved. It follows that for any $\theta\in [0,1]$ and for any index $\mathcal A\in\Omega$, $\hat x^{\theta}_{\mathcal A}\leq z_{\mathcal A}^{\theta}$ (recall the notation $x_i^{\theta}$ defined in Theorem \ref{def1.4}). Then, by the monotonicity of $\rho$, we have
\begin{equation}\label{eq11-29-2}
\rho(\hat x^{\theta})\leq \rho(z^{\theta}).
\end{equation}

Combining \eqref{eq11-29-1}, \eqref{eq11-29-2} with the positive homogeneous property of Lov\'{a}sz extension,
$$
\rho_0(\mathcal S_1)=\hat \rho(\hat x)=\int_{0}^{1} \rho({\hat x}^{\theta})d\theta \leq \int_{0}^{1} \rho({z}^{\theta})d\theta=\hat \rho( z) \leq  \hat \rho({bsx^*})
  =bs\hat \rho(x^*)=bs\cdot opt_{cp}.
$$

Next, consider the second half of the lemma. For each element $e$ with $y^*_e\geq \frac{1}{s}$, by the second constraint of \eqref{eq11-22-3}, and by the observation that $|\Omega_e|\leq b$, we have
\begin{equation}\label{eq11-30-2}
\max_{\mathcal A\in\Omega_e}x^*_{\mathcal A}\geq y_e^*/b\geq 1/bs.
\end{equation}
Hence there is at least one $r_e$-cover $\mathcal A\in\Omega_e$ with value $x^*_{\mathcal A}\geq \frac{1}{bs}$, and thus $\hat x_{\mathcal A}=1$. That is, after the deterministic rounding, at least one $r_e$-cover is chosen into $\mathcal S_1$, and thus $e$ is fully covered.
\end{proof}

\begin{lemma}\label{lem2}
For the collection of sets $\mathcal S_2$ computed by Algorithm \ref{algo1}, the expected cost of $\mathcal S_2$ satisfies $\mathbb{E}[\rho_0(\mathcal S_2)]\leq bs\ln(\frac{s}{s-t})opt_{cp}$.
\end{lemma}
\begin{proof}
Observe that each of the second ``for'' loop of Algorithm \ref{algo1} is in fact a realization of Lov\'{a}sz extension in Theorem \ref{def1.4} (one may refer to \cite{Chekuri}). So the expectation of the cost of those sets in each iteration is $\hat \rho(x^*)=opt_{cp}$. Since $\mathcal S_2$ is the union of these sets, so after $bs\ln(\frac{s}{s-t})$ iterations, $\mathbb{E}[\rho_0(\mathcal S_2)]\leq bs\ln(\frac{s}{s-t})opt_{cp}$.
\end{proof}

In the following, when we say that element $e$ is fully covered by $\mathcal S_2$, it means that the {\em remaining} covering requirement of $e$ is satisfied by $\mathcal S_2$. Using such a convention, we denote by $\mathcal C(\mathcal S_2)$ the set of elements fully covered by $\mathcal S_2$, and let $p(\mathcal S_2)=\sum_{e\in\mathcal C(\mathcal S_2)}p(e)$. Notice that $\mathcal S_2$ is in fact a random sub-collection, and thus $p(\mathcal S_2)$ is a random value. To be more strict, let $\hat y_{e}$ be the random variable which takes value $1$ if $e$ is fully covered by $\mathcal S_2$, and takes value $0$ otherwise. Then \begin{equation}\label{eq12-27-1}
p(\mathcal S_2)=\sum_{e\in \mathcal C(\mathcal S_2)}p(e){\hat y_{e}}.
\end{equation}
The next lemma gives an upper bound for the expected value of $p(\mathcal S_2)$.

\begin{lemma}\label{lem3}
For the collection of sets $\mathcal S_2$ computed by Algorithm \ref{algo1}, the expected profit of $\mathcal S_2$ satisfies $E[p(\mathcal S_2)]\geq tq'P$.
\end{lemma}
\begin{proof}
Since $\mathbb{E}[p(\mathcal S_2)]=\sum_{e\in E'}p(e)Pr[\hat{y}_e=1]$ and
$$
   \sum_{e\in E'}p(e)y^*_e \geq qP-\sum_{e\in \mathcal C(\mathcal S_1)}p(e)y^*_e\geq  qP-\sum_{e\in \mathcal C(\mathcal S_1)}p(e) =q'P,
$$
it suffices to prove that for each $e\in E'$,
 \begin{equation}\label{eq11-30-1}
   Pr[\hat{y}_e=1]\geq ty^*_e.
 \end{equation}
Notice that for each $e\in E'$, $y^*_e\leq 1/s$. Since we have assumed $t<s$, so $ty^*_e<1$. Then, proving \eqref{eq11-30-1} is equivalent to proving
\begin{equation}\label{eq11-30-3}
    Pr[\hat{y}_e=0]\leq 1-ty^*_e.
\end{equation}

In a ``for'' loop with a uniformly randomly chosen $\theta\in[0,1]$, an $r_e$-cover $\mathcal A$ is chosen into $\mathcal S_2$ if and only if $x^*_{\mathcal A}\geq \theta$. For an element $e\in E'$, it is not fully covered by those sets chosen into $\mathcal S_2$ in this ``for'' loop if and only if $\theta>\max\{x^*_{\mathcal A}\colon \mathcal A\in\Omega_e\}$. This occurs with probability $1-\max\limits_{\mathcal A\in \Omega_e} x^*_{\mathcal A}$. Since $\max\limits_{\mathcal A\in \Omega_e} x^*_{\mathcal A}\geq \frac{y^*_e}{b}$ (see \eqref{eq11-30-2}), we have
$$
  1-\max\limits_{\mathcal A\in\Omega_e} x^*_{\mathcal A}\leq 1-\frac{y^*_e}{b}.
$$
So, after $bs\ln(\frac{s}{s-t})$ iterations,
$$
Pr[\hat{y}_e=0] \leq \left(1-\frac{y^*_e}{b}\right)^{bs\ln(\frac{s}{s-t})}\leq e^{-sy^*_e \ln(\frac{s}{s-t})}=\left(\frac{s}{s-t}\right)^{-sy^*_e},
$$
where the second inequality uses the fact that $1-x\leq e^{-x}$. Denote $f(x)=(\frac{s}{s-t})^{-sx}$ and $g(x)=1-tx$. Notice that $f(x)$ is a convex function and $g(x)$ is a linear function. Furthermore, $f(0)=g(0)$, $f(1/s)=g(1/s)$. So $f(x)\leq g(x)$ in interval $[0,1/s]$. Since for each $e\in E'$, $0\leq y^*_e\leq 1/s$. So, $(\frac{s}{s-t})^{-sy^*_e}\leq 1-ty^*_e$. Property \eqref{eq11-30-3} is proved, and the lemma follows.
\end{proof}

\begin{remark}\label{rem1}
{\rm One may be wondering what if $tq'P$ is larger than the profit of those remaining elements which are not fully covered by $\mathcal S_1$. This cannot happen because after the first stage of deterministic rounding, the total profits of remaining elements is $q'P+(1-q)P$. Since it is required that $1<t< s\leq 1/q\leq 1+\frac{1-q}{q'}$, we have $tq'P<q'P+(1-q)P$.}
\end{remark}

Now we will show that by choosing suitable parameters $s$ and $t$, Algorithm \ref{algo1} produces a feasible solution with performance ratio $O(b)$ with high probability.

\begin{Theorem}\label{th2}
Setting $s=1/q$ and $t=1/\sqrt{q}$, Algorithm \ref{alg1} produces a feasible solution to SCPMC with high probability whose cost is $O(b)opt_{cp}$, where $b=\max_e\binom{f}{r_{e}}$.
\end{Theorem}
\begin{proof}
Notice that for the above $s$ and $t$, we have $1<\frac{1}{\sqrt[3]{q}}<t=\frac{1}{\sqrt{q}}<s=\frac{1}{q}.$

The outline of the proof is as follows: we first show that the sum of the probabilities for the following two events is a constant strictly smaller than 1; then a feasible solution with desired performance ratio can be achieved with high probability by repeating Algorithm \ref{algo1}  $O(\ln(n))$ times. The two events are:
\begin{itemize}
\item[$(\romannumeral1)$] $ \rho_0(\mathcal S_2) > bsl\ln(\frac{s}{s-t})opt_{cp}$, where $l=\frac{1-q}{(t-1)\varepsilon}$; %$l=\frac{1}{1-\exp\left(-\frac{\left(\frac{1}{\sqrt[3]{q}}-1\right)^{2}}{\frac{2}{\sqrt[3]{q}}}\right)}$;
\item[$(\romannumeral2)$] $p(\mathcal S_2) < q'P$.
\end{itemize}

For event $(\romannumeral1)$, using Markov inequality and Lemma \ref{lem2}, we have
\begin{equation}\label{eq12-4-1}
Pr\left[\rho_0(\mathcal S_2) > bsl\ln\left(\frac{s}{s-t}\right)opt_{cp}\right]\leq \frac{1}{l}=\frac{(t-1)\varepsilon}{1-q}=1-\frac{(1-q)+(1-t)\varepsilon}{1-q}.
\end{equation}

For event $(\romannumeral2)$, since $q'>\frac{qP-(q-\varepsilon)P}{P}=\varepsilon$ by Algorithm \ref{alg1}, $q'P+(1-q)P-p(\mathcal S_2)\geq 0$ by Remark \ref{rem1}, and $E[q'P+(1-q)P- p(\mathcal S_2)]\leq q'P+(1-q)P-tq'P$ by Lemma \ref{lem3}, using Markov inequality,
\begin{align}\label{eq12-4-2}
  Pr\left[p(\mathcal S_2)\leq q'P\right] & =Pr \left[q'P+(1-q)P-p(\mathcal S_2)\geq (1-q)P\right] \\
  & \leq \frac{q'P+(1-q)P-tq'P}{(1-q)P}<\frac{(1-q)+(1-t)\varepsilon}{1-q}\nonumber.
\end{align}
%For event $(\romannumeral2)$, first notice that $\hat y_{e}$'s are random variables in $\{0,1\}$ which are independent with each other. Second, by the same reason as in Assumption 2, we may assume $p(e)/q'P\leq 1$. Hence if we denote by $a_e=p(e)/ q'P$ and $z(e)=p(e){\hat y_{e}} / q'P$, then $p(\mathcal S_2)/q'P=\sum_{e\in\mathcal C(\mathcal S_2)}p(e){\hat y_{e}} / q'P=\sum_{e\in\mathcal C(\mathcal S_2)}z_e$ is a sum of independent random variables with each random variable $z_e\in \{0,a_e\}$ and $a_e\leq 1$. Furthermore, by Lemma \ref{lem3}, we have $E[p(\mathcal S_2)/q'P]\geq t$. Then, by Chernoff bound, we have
%\begin{equation}\label{eq12-27-2}
% Pr\left[\frac{p(\mathcal S_2)}{q'P}\leq (1-\varepsilon)t\right]< \exp(-\frac{\varepsilon^{2}t}{2}).
%\end{equation}
%Taking $\varepsilon=1-1/t$, then $(1-\varepsilon)t=1$ and the righthand side of \eqref{eq12-27-2} becomes $\exp(-\frac{(t-1)^{2}}{2t})$. Combining these with the observation that function $f(t)=\exp(-\frac{(t-1)^{2}}{2t})$ is monotone decreasing for $t>1$, we have
%\begin{align}\label{eq12-4-2}
%  Pr\left[p(\mathcal S_2)\leq q'P\right] & < \exp\left(-\frac{(t-1)^{2}}{2t}\right)=\exp\left(-\frac{\left(\frac{1}{\sqrt{q}}-1\right)^{2}}{\frac{2}{\sqrt{q}}}\right)\\
%  & <\exp\left(-\frac{\left(\frac{1}{\sqrt[3]{q}}-1\right)^{2}}{\frac{2}{\sqrt[3]{q}}}\right).\nonumber
%\end{align}

Adding inequalities \eqref{eq12-4-1} and \eqref{eq12-4-2}, the probability that either event $(\romannumeral1)$ occurs or event $(\romannumeral2)$ occurs is upper bounded by a constant which is strictly smaller that $1$. Hence, by repeating Algorithm \ref{algo1} $O(\ln(n))$ times, with a high probability, $p(\mathcal S_2)\geq q'P$ and $\rho_0(\mathcal S_2)\leq  bsl\ln(\frac{s}{s-t})opt_{cp}$. Combining these with Lemma \ref{lem1}, with high probability, $p(\mathcal S')=p(\mathcal S_1)+p(\mathcal S_2)\geq qP$ and
$$
   \rho_0(\mathcal S') \leq \rho_0(\mathcal S_1)+\rho_0(\mathcal S_2)\leq bs\left(1 +l\ln\left(\frac{s}{s-t}\right)\right)opt_{cp} = O\left(\frac{b}{\varepsilon}\right)opt_{cp},
$$
where the firs inequality uses \eqref{eq12-28-1} and the constant in big O is$\frac{1}{\sqrt{q}-q}\left((\frac{1}{\sqrt{q}-1})\varepsilon+(1-q)\ln(\frac{1}{1-\sqrt{q}})\right)$. The theorem is proved.
\end{proof}

\section{Approximation Algorithm for SCPSC}

As a corollary of Theorem \ref{th2}, the minimum submodular cost partial set cover problem (SCPSC for short, in which the covering requirement for each element is one) admits a bicriteria randomized $(\varepsilon,O(\frac{f}{\varepsilon}))$-approximation. In the following, we show that an adaptation of our method can yield the same approximation for SCPSC even if the submodular function $\rho_0$ is non-monotone. The idea behind the adaptation is that in this case, a natural constraint is sufficient (we do not need to use the more complicated $r_e$-covers), and thus a technique similar to that in \cite{Iwata} dealing with non-monotone submodular functions can be used.

The SCPSC problem can be modelled as the following integer program:
\begin{align*}
\min\ & \rho_0(x)  \nonumber\\
 s.t.\    \sum_{e:\ e\in E}p_ey_e  & \geq qP, \nonumber\\
 \sum_{S:\ S\in \mathcal S}x_S  & \geq y_e, \ \mbox{for any}\ e\in E\\
   x_S \in  \{0,1\} & \ \ \mbox{for}\ S\in\mathcal S\nonumber\\
   y_e \in  \{0,1\} & \ \ \mbox{for}\ e\in E\nonumber,
\end{align*}
Its relaxation is a convex program:
\begin{align}\label{eq1-16-1}
  \min\ & \hat{\rho_0}(x)  \nonumber\\
 s.t.\    \sum_{e:\ e\in E}p_ey_e  & \geq qP, \nonumber\\
 \sum_{S:\ S\in \mathcal S}x_S  & \geq y_e, \ \mbox{for any}\ e\in E\\
   x_S& \geq 0\ \mbox{for}\ S\in\mathcal S\nonumber\\
  1\geq y_e& \geq 0\ \mbox{for}\ e\in E\nonumber
\end{align}
Notice that since we can use $\rho_0$ as objective function here, the convexity follows directly from the submodularity of $\rho_0$. While for program \eqref{eq11-22-3}, its convexity is guaranteed by Lemma \ref{lem1-15-1}, which is no longer true if $\rho_0$ is non-monotone (see Remark \ref{rem1-15-2}).

Define a new function $\gamma$ by $\gamma(\mathcal S')=\min\{\rho_0(\mathcal S'')\colon \mathcal S'\subseteq \mathcal S''\subseteq \mathcal S\}$. Then $\gamma$ is a nonnegative monotone nondecreasing submodular function (see \cite{Iwata}). For any sub-collection $\mathcal S'\subseteq \mathcal S$, the value of $\gamma(\mathcal S')$ can be determined in polynomial time by an algorithm for submodualr function minimization. Let $\mathcal S_0'$ be the minimizer, that is, $\mathcal S'\subseteq \mathcal S_0'\subseteq \mathcal S$ and $\rho(\mathcal S'_0)=\gamma(\mathcal S')$. It should be noticed that $\mathcal S_0'$ can fully cover all those elements which are fully covered by $\mathcal S'$ (since $\mathcal S'\subseteq \mathcal S'_0$).

Our algorithm for SCPSC is similar to Algorithm \ref{alg1} with the following two differences. First, replace convex program \eqref{eq11-22-3} by \eqref{eq1-16-1}. Second, having obtained $\mathcal S_1$, compute $(\mathcal S_1)_0$ and replace $\mathcal S_1$ by $(\mathcal S_1)_0$ in the remaining part of Algorithm \ref{alg1}.

Notice that in the analysis, monotonicity is used only in Lemma \ref{lem1}. So, to obtain the desired result, we only need to prove the following lemma.

\begin{lemma}\label{lem4}
$\rho((\mathcal S_1)_0)\leq bs\cdot opt_{cp}$.
\end{lemma}
\begin{proof}
Let $\hat x$ be the indicator vector of $\mathcal S_1$. By the monotonicity of $\gamma$, the Lov\'asz extension $\hat{\gamma}$ is also monotone nondecreasing. Hence it follows from $\hat x\leq bsx^*$ that
\begin{equation}\label{eq1-16-2}
\hat \gamma(\hat x)\leq\hat \gamma(bsx^*).
\end{equation}
For any sub-collection $\mathcal S'\subseteq \mathcal S$, by Definition \ref{def1.3}, $\gamma(\mathcal S')\leq \rho(\mathcal S')$. By the definition of Lov\'asz extension in Definition \ref{def1.3}, we have
\begin{equation}\label{eq1-16-3}
\hat \gamma(x)\leq \hat \rho(x) \ \mbox{holds for any vector}\ x\in [0,1]^{|\mathcal S|}.
\end{equation}
Combining \eqref{eq1-16-2}, \eqref{eq1-16-3} with the positive homogeneous property of Lov\'{a}sz extension,
$$\rho((\mathcal S_1)_0)=\gamma(\mathcal S_1)=\hat {\gamma}(\hat x)\leq \hat {\gamma}(bsx^*)=bs\hat {\gamma}(x^*)\leq bs\hat {\rho}(x^*)= bs\cdot opt_{cp}.$$
The lemma is proved.
\end{proof}

From the above argument, we have the following result.

\begin{Theorem}\label{th3}
For any nonnegative submodular function, the $SCPSC$ problem has a bicriteria randomized $(\varepsilon,O(\frac{f}{\varepsilon}))$-approximation with high probability.
\end{Theorem}

%To be more concrete, if $\rho$ is nonnegative and submodular but not monotone nondecreasing, we can replace $\rho$ by let $\mathcal S_1$ be the corresponding collection where $x_S$ rounded up to $1$ in the first deterministic rounding stage.
%For the algorithm part, instead of computing $\mathcal S_1$, we compute $\mathcal S'_1$ such that $\rho(\mathcal S'_1)=\min\{\rho(\mathcal S')| \mathcal S_1\subseteq \mathcal S'\subseteq \mathcal S \}$. This step can be done by finding minimum submodualr function in polynomial time. We call this new algorithm as modified algorithm.
%
%For the analysis part, by deterministic rounding analysis in paper \cite{Iwata}, we can get similar result like Lemma \ref{lem1}. Finally, we have the theorem as follows:

\section{Conclusion}

By introducing a novel convex program describing the minimum submodular cost partial multi-cover problem (SCPMC), we give a randomized $(\varepsilon,O(\frac{b}{\varepsilon}))$-approximation algorithm for SCPMC, where $b=\max_e\binom{f}{r_{e}}$. Since PMC is a special case of SCPMC, the PMC problem also has a bicriteria randomized $(\varepsilon,O(\frac{b}{\varepsilon}))$-approximation algorithm with a high probability. We show that in the case when the covering requirement for each element is one, monotonicity requirement can be dropped off from the cost function. It should be noticed that if we only care about an expected result, then we may obtain a randomized algorithm producing a sub-collection $\mathcal S'$ with $\mathbb E[\rho_0(\mathcal S')]\leq bs(1+\ln\frac{s}{s-t})opt$ and $\mathbb E[p(\mathcal S')]\geq qP$. This can be achieved by modifying $(q-\varepsilon)P$ in Line 8 of Algorithm \ref{algo1} into $qP$.

One question is can one obtain the same result for SCPMC without monotonicity requirement? Another question is what if $r_{\max}$ is not upper bounded by a constant?

\section*{Acknowledgements}
This research is supported by NSFC (11531011, 61222201).

\end{document}